\documentclass[10pt]{article}

\usepackage{latexsym}
\usepackage{amsfonts, amsmath}
\usepackage{amssymb}
\usepackage{theorem}
\usepackage[english]{babel}
\usepackage{longtable}

\usepackage[margin=3.4cm]{geometry}

\usepackage{color}

\newtheorem{theorem}{Theorem}[section]

\newtheorem{corollary}[theorem]{Corollary}
\newtheorem{lemma}[theorem]{Lemma}
\newtheorem{conjecture}[theorem]{Conjecture}
\newtheorem{example}[theorem]{Example}
\newtheorem{question}[theorem]{Question}

\def\ND{\succeq\kern-10pt/\kern5pt}
\def\Ndelta{\delta\kern-5pt/\kern3pt}

\newenvironment{proof}
{\begin{trivlist}\item[]{Proof:}}{\hfill{$\square$}\noindent\end{trivlist}}

\begin{document}

\title{Characterization of threshold functions: state of the art, some new contributions and open problems}

\author{\normalsize{
Josep Freixas\footnote{
Universitat Polit\`{e}cnica de Catalunya (Campus Manresa), in the Department of  Mathematics;
Av. Bases de Manresa, 61-73, E-08242 Manresa, Spain.}, \
Marc Freixas\footnote{Industrial Engineer working for Cirprotec, E-08233 Terrassa, Spain.} \ and \
Sascha Kurz\footnote{
University of Bayreuth, in the Department of  Mathematics; 95440 Bayreuth, Germany.}
}}
\date{\small{}}
\maketitle
\begin{abstract}
This paper has a twofold scope. The first one is to clarify and put in evidence the isomorphic character of two theories developed in quite different fields: on one side, threshold 
logic, on the other side, simple games. One of the main purposes in both theories is to determine when a simple game is representable as a weighted game, which allows a very 
compact and easily comprehensible representation. Deep results were found in 
threshold logic in the sixties and seventies for this problem. However, game theory has taken the lead 
and some new results have been obtained for the problem in the last 
two decades. The second and main goal of this paper is to provide some new results on this problem and propose several open questions and conjectures 
for future research.
The results we obtain depend on two significant parameters of the game: the number of types of equivalent players and the number of types of shift-minimal 
winning coalitions.

\medskip

\end{abstract}

\section{Introduction}  \label{SECTIONintroduction}

The study of switching functions goes back at least to Dedekind's 1897 work \cite{dedekind}, in which he determined the exact number of simple games with 
four or fewer players. Since that time these structures have been investigated in a variety of different contexts either 
theoretically~\cite{Gol59, HIP81, HaHo92, HKR00, BoBr03} in the context of Boolean functions or because of their numerous applications: 
neural networks~\cite{AnHo94}, simple games~\cite{vNMo44, Isb56, Isb58, Pel68}, threshold logic~\cite{Elg60, Cho61, Gab61, Hu65, Mur71}, 
hypergraphs~\cite{RRST85}, coherent structures~\cite{Ram90},
learning theory~\cite{Litt88}, complexity theory~\cite{BeWe06}, and secret sharing~\cite{Sim90, Tas07, BTW08}. Several books on neural networks 
have studied these structures:~\cite{Par94, RSO94, SRK95, Pic00}.

Logic gates, switching functions or Boolean functions can be thought of as simple games, with weighted games playing the role of threshold functions. 
To the best of our knowledge the first work linking threshold logic and simple games is due to Dubey and Shapley~\cite{DuSh79} and a compact study 
encompassing knowledge in both fields is due to Taylor and Zwicker~\cite{TaZw99}. 

As an example for a switching function or a simple game one may consider the process of coordination of the weekend activities of a family. Assume that 
the family consists of the parents Ann and Bob and their children Claire and Dylan. A proposal is accepted if at least one of the parents and at least one 
of the children agrees, while each person can either agree or disagree. The underlying decision rule can be modeled as a simple game.\footnote{The minimal 
winning coalitions are given by $\{A,C\}$, $\{A,D\}$, $\{B,C\}$, and $\{B,D\}$, see Section~\ref{sec_terminology} for the definitions.} A compact way to 
represent a simple game is by using weights for each player such that a proposal is accepted if and only if the weight sum of its supporters meets or exceeds 
a given quota (or threshold). If such a representation exists, the simple game is called a weighted game. In our example no weighted representation exists, since 
the coalitions of the parents and of the children cannot push through a proposal, while they can if they split differently in coalitions of size 
two.\footnote{Using the notation from Section~\ref{sec_characterization}, $\langle \{A,C\},\{B,D\}\,\Vert\,\{A,B\},\{C,D\}\rangle$ is a trading transform, which certifies 
non-weightedness.} However, every simple game can be written as the intersection of some weighted games. The Lisbon voting rules of the EU Council provide 
a non-weighted real--world example where quite a few weighted games are need in such a representation, see \cite{KuNa16} for the details.

One of the most fundamental questions in all of the above mentioned areas is to 
characterize which monotonic switching functions (simple games) are weighted threshold functions (weighted games). In threshold logic this is known 
as the linear separability problem. This question has also been posed in other research fields by using different terminologies, which are essentially 
equivalent. Three different treatments to solve this problem have been considered.

The first consists in studying the consistency of a system of inequalities. Each inequality is formed by the inner product of two vectors: a non-negative 
integer vector of weights which represents the unknown variables and the vector formed by the subtraction of a true vector (winning coalition) minus a false 
vector (losing coalition). The system is formed by considering all possible subtractions of true and false vectors. If the switching function is a threshold 
function then each inequality must be positive and the system of inequalities is consistent. A theorem on the existence of solutions for systems of linear 
inequalities was given in~\cite{Chv83}. Linear programming is also a useful tool as shown in~\cite{KuTa13, FrKu14MSS}.

The second treatment, very close to the previous one, is a geometric approach based on the existence of a separating hyperplane that separates true vectors 
from false vectors. This procedure is elegant but not very efficient in practice. A use of the geometrical approach can be found in~\cite{EiLe89}. Reference \cite{HoZw14} 
proposes a variant of it.

The idea behind the third approach lies in the consideration of exchanges among vectors and the possibility to convert some true vectors into false vectors, 
no matter the number of vectors involved in these exchanges. The early works of~\cite{Elg60} and~\cite{Cho61}, reexamined for simple games in~\cite{TaZw92}, 
are the central point of this work. The class of threshold functions admits a structural characterization, the asummability property,  that is both natural 
and elegant. Some of the deepest results on this subject were obtained in the area of threshold logic during two decades from the fifties to the seventies by 
Chow, Elgot, Gabelman and Winder, as reported by Hu~\cite{Hu65}, and continued by Muroga~\cite{Mur71} and Muroga et al.~\cite{MTT61, MTK62, MTB70}.

The interpretation of Taylor and Zwicker for the asummability condition in terms of trades among coalitions in~\cite{TaZw92} together with the work in~\cite{TaZw95} 
stimulated the interest for the problem of characterizing weighted games within simple games. In their book~\cite{TaZw99} they adapted, for simple games, the most 
important results of threshold logic in relation to the linear separability problem. In particular, their property of trade-robustness is equivalent to the  
property of asummability. However, trade-robustness is more transparent in the theoretical context of voting since it gives rise to some intuitions concerning the 
idea of trading players among coalitions.
Freixas and Molinero~\cite{FrMo09DAM} propose a relaxation of trade-robustness for complete simple games and called it invariant-trade robustness, which is 
less costly in terms of the lenght of the corresponding certificates. 
In this paper we deduce some new results using this property.

As mentioned before, each simple can can be represented as an intersection of weighted games, which allows a compact representation once the number of required 
weighted games is small. In general this number, called dimension of the simple game, can be large, see e.g.~\cite{OlKuMo16}, where the worst case asymptotics has been 
determined via a connection to coding theory. If the dimension is small, e.g.~power indices can in general be more simply and efficiently computed. Here we treat 
the extreme case of dimension $1$, i.e., we consider the relavant issue whether a given simple game (switching function) is weighted and propose some new 
characterizations.

The organization of the paper is as follows.
The necessary basic terminology of simple games is reviewed in Section~\ref{sec_terminology}.
Section~\ref{sec_characterization} recalls the main general results on the characterization of weighted games within the class of simple games,
and it identifies the analogue terminologies used in threshold logic and simple games.
The problem of the characterization of weighted games can be restricted to the class of complete games, a parametrization result for classifying them, 
up to isomorphisms, which will be
intensively used in the next sections, is recalled in Section 4.
Sections 5, 6 and 7 provide new results on the characterization of weighted games. 
Cases for which $2$-invariant trade robustness is conclusive are presented in Section 5.  
$m$-invariant trade robustness is studied in Section 6, while Section 7 gives some numerical and experimental data. 
Several questions and conjectures are proposed for future research in Section 8.
In Section 9 we draw a conclusion. 

\section{Terminology in the context of voting simple games}
\label{sec_terminology}

Simple games or binary voting systems can be viewed as models of voting systems in which a
single alternative, such as a bill or an amendment, is pitted
against the status quo. A \emph{simple game} $G$ is a pair $(N, \mathcal W)$ in
which $N = \{1,2,\dots,n\}$ is the set of players or voters and $\mathcal W$ is a collection of
subsets of $N$ that satisfies: \emph{(1)} $N \in \mathcal W$,
\emph{(2)} $\emptyset \notin \mathcal W$ and \emph{(3)} the
\emph{monotonicity} property: $S \in \mathcal W$ and $S \subseteq T
\subseteq N$ implies $T \in \mathcal W$.

Any set of voters is called a \emph{coalition}, and the set $N$ is
called the \emph{grand coalition}. Members of $N$ are called
\emph{players} or \textit{voters}, and the subsets of $N$ that are
in $\mathcal W$ are called \emph{winning coalitions}.

The intuition here is that a set $S$ is a winning coalition \emph{if and
only if} the bill or amendment passes when the players in $S$ are
precisely the ones who voted for it. A subset of $N$ that is not in
$\mathcal W$ is called a \emph{losing coalition}. A
\emph{minimal winning coalition} is a winning coalition all of
whose proper subsets are losing. A
\emph{maximal losing coalition} is a losing coalition all of
whose proper supersets are winning. Because of monotonicity, any simple
game is completely determined by its set of minimal winning
coalitions, which is denoted by $\mathcal W^m$ or by its set of maximal losing coalitions,
which is denoted here
$\mathcal L^M$. A voter $a \in N$ is
\emph{null} if $a$ does not belong to any minimal winning
coalition. A player $a \in N$ has \emph{veto} if $a$ belongs to all winning coalitions.

Before proceeding, we present 
two real-world examples of simple
games (see Taylor and Pacelli~\cite{TaPa08} for a thorough presentation
of these two examples).

\begin{example} \label{EXAMPLE:SC} The United Nations Security Council. The
voters in this system are the fifteen countries that make up the
Security Council, five of which are permanent members whereas the
other ten are non-permanent members. Passage requires a total of at
least nine of the fifteen possible votes, subject to a veto due to a
no vote from any one of the five permanent members. This model
ignores abstention. For a treatment of this example considering the
possibility of abstention we refer the reader to \cite{FrZw03}.
\end{example}

\begin{example} \label{EXAMPLE:CC} The System to amend the Canadian Constitution.
Since 1982, an amendment to the Canadian Constitution can become law
only if it is approved by at least seven out of the ten Canadian
provinces, subject to the proviso that the approving provinces have,
among them, at least half of Canada's population. It was first
studied in Kilgour~\cite{Kil83}. An old census (in percentages) for the
Canadian provinces was: 1. Ontario $(34\%)$, 2. Quebec $(29\%)$, 3. British Columbia
$(9\%)$,  4. Alberta $(7\%)$, 5. Saskatchewan $(5\%)$, 6. Manitoba $(5\%)$,
7. Nova Scotia $(4\%)$, 8. New Brunswick $(3\%)$, 9. Newfoundland
$(3\%)$, 10. Prince Edward Island $(1\%)$.

For example observe that coalitions (from now on we use abridgments
to denote the provinces):
\newline
$X_1 = \{Que, BC, Alb, Sas, Man, NS, NB\}$ and
$X_2 = \{Ont, Sas, Man, NS, NB, Newf, PEI \}$ are minimal
winning coalitions because they both have exactly $7$ provinces and
their total population surpasses the $50\%$. Instead, coalitions:
\newline
$Y_1 = \{Ont, Que, Sas, Man, NS, NB \}$ and
$Y_2 = \{BC, Alb, Sas, Man, NS, NB, Newf, PEI \}$
are both losing because $Y_1$ does not
have $7$ or more members and $Y_2$ does not reach the $50\%$ of the
total Canada's population.
\end{example}

A fundamental subclass of simple games are the classes of weighted simple
games and complete simple games. 
A simple game $G=(N, \mathcal W)$ is said to be
\emph{weighted} if there exists a ``weight function" $w:N
\rightarrow \mathbb{R}_{\ge 0}$ and a ``quota" $q \in \mathbb{R}_{>0}$ such that
a coalition $S$ is winning precisely when the sum of the weights of
the players in $S$ meets or exceeds the quota. Any specific example of such a weight function $w: N \rightarrow
\mathbb{R}$ and quota $q$ are said to
\emph{realize} $G$ as a weighted game. A particular realization of
a weighted simple game is denoted as $[q;w_1, \dots, w_n]$.

For instance, $[k;\overbrace{1, \dots, 1}^n ]$ for some $k=1,
\dots,n$ is a feasible realization for a weighted game in which all players are symmetric;
here the
game is called a \emph{$k$--out--of--$n$ simple game}. A realization
of Example \ref{EXAMPLE:SC} is
$[39;7,7,7,7,7,1,1,1,1,1,1,1,1,1,1]$, where $7$ is the weight for a
permanent member and $1$ the weight for a non-permanent member. Instead, Example
\ref{EXAMPLE:CC} cannot be represented as a weighted game. Indeed, if the game was weighted we would have $w(X_1)>w(Y_1)$ and $w(X_2)>w(Y_2)$, i.e., 
$$
  w_2+(w_3+\dots+w_8) 
  \,>\, 
  (w_1+w_2)+(w_5+\dots+w_8) \quad\text{and}\quad 
  w_1+(w_5+\dots+w_{10}) 
  \,>\, 
  (w_3+\dots+w_{10}).
$$
After simplification we obtain $w_3+w_4>w_1$ and $w_1>w_3+w_4$, which 
is a contradiction. In these inequalities, $w_1$ represents, the weight for Ontario, the most populated province; $w_2$ represents, 
the weight for Quebec, the second most populated province; and so on.

It is quite intuitive to observe that a permanent
member has more influence than a non-permanent member in the voting systems described in Example~\ref{EXAMPLE:SC}.
The same occurs in Example~\ref{EXAMPLE:CC} where any of the two big provinces are more influential  
than any other of the remaining eight provinces. 
The ``desirability relation" represents a way to make the idea,
that a particular voting system may give to one voter more influence
than another, more precise. Isbell already used it in \cite{Isb58}.

Let $G=(N, \mathcal W)$
be a simple game, $a$ and $b$ be two voters. Player $a$ is said to
be \emph{at least as desirable as} $b$ as coalitional partner if for
every coalition $S$ such that $a \notin S$ and $b \notin
S$, $S \cup \{b\} \in \mathcal W$ implies $S \cup
\{a\} \in \mathcal W$. If moreover, there exists a coalition $T$ such that $a \notin T$ and $b
\notin T$, $T \cup \{a\} \in \mathcal W$ and $T \cup \{b\} \notin
\mathcal W$, then $a$ is (strictly) \emph{more desirable} than $b$.
Finally, $a$ and $b$ are said to be \emph{equally
desirable} if $a$ is at least as desirable as $b$ and the converse is also true.
The notations $a \succsim b$, $a \succ b$ and $a \sim b$ respectively stand for:
$a$ is at least as desirable as $b$, $a$ is strictly more desirable than $b$,
and $a$ and $b$ are equally desirable.  It is straightforward to check that $\sim$ is an
equivalence relation, and that the desirability relation $\succsim$ is a partial ordering of
the resulting equivalence classes.

A simple game $G=(N,\mathcal W)$
is \emph{complete (or linear)} if the desirability relation is a
complete preordering. Note that every weighted game is complete, since
for any realization it holds that $w_a \geq w_b$ implies $a \succsim b$.
But the converse is not true as Example~\ref{EXAMPLE:CC} shows.

In a complete simple game we may decompose $N$ into 
a collection of subsets, called classes, $N_1 > N_2
> \dots > N_t$ forming a partition of $N$.
Those classe should be the equivalence classes ordered by desirability, i.e., 
if $a \in N_p$ and $b \in N_q$ then: $p=q$ if and only if $a \sim b$
and, $p<q$ if and only if $a \succ b$.
Two parameters are of fundamental importance 
in our study.
One of them is the number of equivalence classes, $t$, in a complete game, i.e., 
a measurement of heterogeneity. 
In Example~\ref{EXAMPLE:SC} we have $N_1 > N_2$ where $N_1$ is formed by
the five permanent members and $N_2$ for the non-permanent ones, while in
Example \ref{EXAMPLE:CC} we have $N_1 > N_2$, where $N_1$ is formed by
the two big provinces and $N_2$ for the other eight provinces. Thus, in both examples $t=2$.

To define the second fundamental parameter for our study we need another definition. Given a simple game, a 
\emph{shift-minimal winning coalition} $S$ is a minimal winning coalition such that
$(S \setminus \{a\}) \cup \{b\}$ is losing whenever $a \succ b$ with $a \in S$ and $b \notin S$. Note that
a coalition of seven members in Example \ref{EXAMPLE:CC} containing both, Quebec and Ontario, is a minimal winning coalition 
but it is not shift-minimal winning, since a replacement of a big province in the coalition for a province not belonging 
to the coalition still leaves the new coalition winning. Analogously, a \emph{shift-maximal losing coalition} $S$ is a maximal losing coalition such that
$(S \setminus \{b\}) \cup \{a\}$ is winning whenever $a \succ b$ with $b \in S$ and $a \notin S$.

Two shift-minimal winning coalitions, $S$ and $T$, are said to be \emph{equivalent} coalitions
if $T$ can be obtained from $S$ by any sequence of one--to--one exchanges of equally desirable voters. If the game is 
complete we can consider the parameter $r$ which is the maximal number of non-equivalent shift-minimal winning coalitions.

Observe that the complete game from Example~\ref{EXAMPLE:SC} has ${10 \choose 4}=210$ minimal winning coalitions which are 
also shift-minimal winning coalitions, each consisting of all five permanent members and four arbitrary non-permanent members. 
However, all of them are equivalent in the previous sense. Thus, for Example~\ref{EXAMPLE:SC} the two  parameters we lay stress on are $r=1$ and $t=2$. 

The simple game from Example~\ref{EXAMPLE:CC} has $112$ minimal winning coalitions, $56$ of them formed by one of the two big provinces and 
six other provinces, which are also shift-minimal winning coalitions. The game has $56$ additional winning coalitions formed by the two big 
provinces and five other provinces, but these are not shift-minimal winning coalitions. The $56$ shift-minimal winning coalitions are equivalent 
among them in the previous sense. Thus, for Example~\ref{EXAMPLE:CC} the two parameters we lay stress on again are $r=1$ and $t=2$. 
The case $r=1$ is considered in Subsection 5.1.

\section{Some results on the characterization of weighted games}
\label{sec_characterization}

We introduce a notion of trades among coalitions, which is natural in game theory and in economic applications, see~\cite{TaZw99} for 
motivating examples. Suppose $G=(N, \mathcal W)$ is a simple game. Then a \emph{trading transform} is a coalition sequence 
$\langle X_1, \dots,X_k | Y_1, \dots,Y_k \rangle$ of even length satisfying the following condition: 
$$| \{ i : \, a \in X_i \}| = |\{ i: \, a \in Y_i \}| \quad \mbox{for all} \; a \in N.$$
The $X$s  are called the pre-trade coalitions and the $Y$s are called the post-trade coalitions.
A \emph{$k$-trade} for a simple game $G$ is a trading transform  $\langle X_1, \dots,X_j | Y_1, \dots,Y_j \rangle$ with $j \leq k$. 
The simple game $G$ is \emph{$k$-trade robust} if there is no trading transform for which all the $X$s are winning in $G$ and all the $Y$s 
are losing in $G$. If $G$ is $k$-trade robust for all $k\in\mathbb{N}$, then $G$ is said to be \emph{trade robust}.

Loosely speaking, $G$ is $k$-trade robust if a sequence of $k$ or fewer (not necessarily distinct) winning coalitions can never be rendered losing by a trade.

\begin{theorem} \label{T:ChowElgotTrade}
(Theorem 2.4.2 in~\cite{TaZw99}, see also~\cite{TaZw92}) Let $G=(N, \mathcal W)$ be a simple game. Then, $G$ is weighted if and only if 
$G$ is trade robust.
\end{theorem}
This result is equivalent to the one given by Elgot~\cite{Elg60} and Chow~\cite{Cho61} in threshold logic. Instead the notation of trade robustness, these authors 
used an equivalent condition of asummability of vectors. If we are restricted to complete simple games and only allow pre-trades of shift-minimal winning coalitions, 
then we may refer to the property of \emph{invariant-trade robustness} instead of trade robustness and Theorem~\ref{T:ChowElgotTrade} can be reformulated in an equivalent way.

\begin{theorem} \label{T:shiftTrade}
(Theorem 4.7 in~\cite{FrMo09DAM}) Let $G=(N, \mathcal W)$ be a complete simple game. Then, $G$ is weighted if and only if 
$G$ is invariant-trade robust.
\end{theorem}

As seen in Example~\ref{EXAMPLE:CC} the trading transform $ \langle X_1,X_2 | Y_1,Y_2 \rangle $ certifies a failure of $2$-invariant-trade robustness and 
therefore this complete simple game is not weighted. It is also trivial to see that the simple game described in Example~\ref{EXAMPLE:SC} 
is invariant-trade robust and therefore weighted.

Suppose $G=(N, \mathcal W)$ is a simple game. Then, 
$G$ is said to be \emph{swap robust} if a one--for--one exchange
between two winning coalitions can never render both losing.
Thus, swap robustness differs from trade robustness in two ways: the
trades involve only two coalitions, and the exchanges are
one--for--one. That is to say, swap robustness considers $m$-trades
of the following type: $m=2$ and $\langle X_1,X_2 | X_1 \setminus \{a\},X_2 \cup \{a\} \rangle$ with $a \in X_1$ and $a \notin X_2$.
It is fairly easy to generate simple games that are not swap robust. The following theorem is a characterization of complete simple games.

\begin{theorem} \label{THEOREM:Swap} (Proposition 3.2.6 in~\cite{TaZw99}) $G$ is a complete simple game \emph{if and only if}
$G$ is swap robust.
\end{theorem}

Clearly, non--complete games are not $2$-invariant trade robust. In fact, it is always possible to find two shift-minimal winning coalitions 
which convert into losing coalitions after a one--for--one exchange. Thus, the difficulty of the problem of determining when a given simple game 
is weighted can be focused exclusively on the class of complete games. To obtain significant results it is helpful to have a compact and manageable 
presentation of these games.

If, in a complete game, the number of the equivalence classes is lower than the number of players, i.e., $t<|N|$, we have such a presentation. 
Indeed, Carreras and Freixas \cite{CaFr96} provide a
classification theorem for complete simple games, here Theorem~\ref{theo:CF96},
that allows to enumerate all these games up to
isomorphism by listing the possible values of certain invariants. An
advantage of using the classification theorem is that it usually
allows to work with a smaller number of vectors than would be
required with minimal winning coalitions. 
In the next section we introduce a notion of trade-robustness based on these invariants.

As the basic game theoretic notions for simple games, we use in this paper, have already been introduced, we list a
list of language analogies between these notions to the fields of
threshold logic or Boolean algebra 
in next subsection.
These analogies allow the easy translation of the results from one field to the other.
In particular, this list will be useful for scholars in threshold logic to be aware of
the new results we find in this paper and the questions and conjectures we propose
to be studied.

\subsection{A list of analogies in the context of threshold logic}

For the sake of simplicity, clarity, and for being coherent with the historical 
studies we 
write the notions in the language of Boolean algebra (very similar to that of neural networks or threshold logic).
Tables~\ref{T:Table00} and~\ref{T:Table01} contain the main equivalences. The list is not exhaustive. Throughout the rest of the paper we exclusively deal with simple games and refer to these two tables for direct analogies of the results we find and the questions and conjectures we pose.

\begin{table}[!htp]
  \caption{Variables and vectors versus players and coalitions.} \label{T:Table00}
  \vskip 0.2truecm
   \centering
\begin{tabular}{r|l}
  variable or node& player or voter  \\
  irrelevant variable & null player \\
  essential variable & vetoer \\
  vector & coalition \\
  true vector &  winning coalition \\
  false vector & losing coalition \\
 minimal true vector &  minimal winning coalition \\
 maximal false vector& maximal losing coalition\\
 shift-minimal true vector &  shift-minimal winning coalition
\end{tabular}
\end{table}

\begin{table}[!htp]
  \caption{Types of functions versus types of simple games.}\label{T:Table01}
  \vskip 0.2truecm
   \centering
\begin{tabular}{r|l}
  switching function & non-monotonic (simple) game  \\
  monotonic switching function & (monotonic) simple game \\
  threshold function & weighted game \\
  k-out-of-n switching function & k-out-of-n simple game \\
  regular function & complete game \\
  k-summable &  not k-trade robust \\
  k-asummable & k-trade robust \\
  k-invariant summable &  not k-invariant trade robust \\
  k-invariant asummable & k-invariant trade robust
\end{tabular}
\end{table}


\section{Symmetries and a parametrization of complete simple games}

For $t<|N|$ types of voters we can represent coalitions in a more compact way. Let $(N,\mathcal{W})$ be a simple game and $N_1,\dots,N_t$ be a partition of the
  player set into $t$ equivalence classes of voters cf.~Section~\ref{sec_terminology}. A \emph{coalition type} (or \emph{coalition vector}) is
  a vector $\overline{s}=(s_1,\dots,s_t)\in (\mathbb{N} \cup \{0\})^t$ with $0\le s_i\le |N_i|$ for all $1\le i\le t$.
  We say that a coalition $S\subseteq N$ has type $\overline{s}$ if $s_i=|S\cap N_i|$ for all $1\le i\le t$.
  A coalition type $\overline{s}$ is called winning if the coalitions of that type are winning. Analogously, the notions
  of minimal winning, shift-minimal winning, losing, maximal losing and shift-maximal losing are translated similarly for coalitional types. So, the simple game from Example~\ref{EXAMPLE:SC} can be described by the unique minimal winning coalition type $(5,4)$ which represents all coalitions with $5$ permanent members and $4$ non-permanent members, and the simple game from Example~\ref{EXAMPLE:CC} can be described by the unique minimal winning coalition type $(1,6)$ which represents all coalitions with $1$ big province and $6$ small provinces.

The notion of a trading transform for coalitions can be transferred to coalitional types for vectors.

\subsection{Coalitional types}

Let $G=(N, \mathcal W)$ be a simple game and $N_1,\dots,N_t$
be a partition into $t$ equivalence classes of players. A
\emph{vectorial trading transform} for $G$ is a sequence
$ \langle \overline{x}_1, \dots, \overline{x}_j;
\overline{y}_1, \dots, \overline{y}_j\rangle$
of coalition types of even length such that
\begin{equation} \label{E:vect-trading}
  \sum_{i=1}^j x_{i,k} =
  \sum_{i=1}^j y_{i,k} \quad  \text{for all} \; \; 1\le k\le t.
\end{equation}

The definition of a vectorial trading transform means that for each component $1\le k\le t$,
the sum of the $k^{th}$ $\overline x$s components coincides with the
sum of the $k^{th}$ $\overline y$s components.

A \emph{vectorial $m$-trade} is a vectorial trading transform with $j\le m$ such that the
$\overline{x}_i$s are winning and after trades, as described in~\ref{E:vect-trading}, convert into $\overline{y}_i$s.

A given $m$-trade can easily be converted into a vectorial $m$-trade. The following lemma shows that the converse is also true, i.e., 
each given vectorial $m$-trade can be converted into an $m$-trade.

\begin{lemma}
  \label{LVTOC}
  For each pair of vectors $\overline{a}=(a_1,\dots,a_r)\in \mathbb{N}_{>0}^r$,
  $\overline{b}=(b_1,\dots,b_s)\in \mathbb{N}_{>0}^s$ with $\sum_{i=1}^r a_i=\sum_{i=1}^s b_i$
  and $m=\max\left(\max_i a_i,\max_i b_i\right)$ there exist two sequences of sets
  $A_1,\dots,A_r\subseteq \{1,\dots,m\}$ and $B_1,\dots,B_s\subseteq \{1,\dots,m\}$
  with $|A_i|=a_i$, $|B_i|=b_i$ and
  $$
    \left|\{i\,:\, j\in A_i\}\right|=\left|\{i\,:\, j\in B_i\}\right|
  $$
  for all $j\in\{1,\dots,m\}$.
\end{lemma}
\begin{proof}
  W.l.o.g.\ we assume $a_1\ge \dots\ge a_r$ and $b_1\ge \dots \ge b_s$. We prove the
  statement by induction on $\sigma=\sum_{i=1}^r a_i$. For $\sigma=1$ we have $r=s=a_1=b_1=m=1$
  and can choose $A_1=B_1=\{1\}$. We remark that the statement is also true for $\sigma=0$, i.e.,
  where $r=s=0$.

  If there exist indices $i,j$ with $a_i=b_j$, then we can choose $A_i=\{1,\dots,a_i\}$,
  $B_j=\{1,\dots,b_j=a_i\}$ and apply the induction hypothesis on $(a_1,\dots,a_{i-1},a_{i+1},$ $\dots,a_r)$
  and $(b_1,\dots,b_{j-1},b_{j+1},\dots,b_s)$.

  In the remaining cases we assume w.l.o.g.\ $a_1=m$ and $b_1<m$. Now let $l$ be the maximal index with $a_l=m$.
  Since $\sum_{i=1}^r a_i=\sum_{i=1}^s b_i$ we have $s\ge l$. So, we can consider the reduction to
  $(a_1-1,\dots,a_l-1,a_{l+1},\dots,a_r)$ and $(b_1-1,\dots,b_l-1,b_{l+1},\dots,b_s)$, where we possibly have to
  remove some zero entries and the maximum entry decreases to $m-1$. Let $A'_1,\dots,A'_r,B'_1,\dots,B'_s
  \subseteq \{1,\dots,m-1\}$ be suitable coalitions (allowing $A'_i=\emptyset$ or $B'_i=\emptyset$ for the ease
  of notation). Adding player~$m$ to the first $l$ coalitions in both cases yields the desired sequences of
  coalitions.
\end{proof}

The construction in Lemma~\ref{LVTOC} for each equivalence class of voters separately converts
a vectorial $m$-trade into an $m$-trade. Also for vectorial $m$-trades we may assume that the winning
coalition types are minimal winning or that the losing coalition types are maximal losing. Since the number
of coalition types is at most as large as the number of coalitions we can computationally benefit from
considering vectorial $m$-trades if the number of types of voters is less than the number of voters.

\subsection{A parametrization of complete simple games}

In a complete simple game $G=(N,\mathcal{W})$ we have a strict ordering between two voters of
different equivalence classes. This ordering entails a hierarchy among voters. Some studies on allowable hierarchies can be found in~\cite{BFP08, FrPo10TD,FMP06}.  
As before, we denote by $N_1 > \dots  > N_t$ the equivalence classes which form the unique partition of $N$ where $a \succ b$ for all $a \in N_i$ and $b \in N_j$ 
with $i<j$.  Let $\overline{n} = (n_1, \dots,n_t)$ where $n_i = |N_i|$ for all $i=1, \dots,t$. Consider
\[
\Lambda (\overline{n})=\{\overline{s}\in (\Bbb{N}\cup \{0\})^{t}:\overline{n}%
\geq \overline{s}\},
\]
where $\geq $ stands for the ordinary componentwise ordering, that is,
$\overline{a}\geq \overline{b}\mbox{ }\mbox{ if and only if }\mbox{ }%
a_{k}\geq b_{k}\mbox{
for every }k=1,...,t.
$
and also consider the weaker ordering $\succeq$ given by comparison
of partial sums, that is,
\[
\overline{a}\,\succeq \,\overline{b}\;\mbox{if and only if}\mbox \;
\sum_{i=1}^k a_i \geq \sum_{i=1}^k b_i  \; \mbox{for} \; k=1,..,t.
\]
If $\overline{a}\,  \succeq \,\overline{b}$ we say that $\overline{a}$
\textit{dominates} $\overline{b}.$

The couple $(\Lambda (\overline{n}),\succeq )$ is a distributive lattice and
possesses a maximum (respectively, minimum) element, namely $\overline{n}%
=(n_{1},...,n_{t})$ (resp. $\overline{0}=(0,...,0)$).
  As abbreviations we use $\overline{a}\succ \overline{b}$ for the cases where $\overline{a}\succeq \overline{b}$
  but $\overline{a}\neq \overline{b}$ and $\overline{a}\bowtie\overline{b}$ for the cases where neither
  $\overline{a}\succeq \overline{b}$ nor $\overline{b}\succeq \overline{a}$.

The interpretation of $\overline{a}\succeq \overline{b}$ is as follows. If $\overline{b}$ is a winning coalitional vector and
$\overline{a}\succeq \overline{b}$, then also $\overline{a}$ is winning. Similarly, if $\overline{a}$ is
losing then $\overline{b}$ is losing too for all $\overline{a}\succeq \overline{b}$.

A winning coalitional vector
$\overline{a}$ such that $\overline{b}$ is losing for all $\overline{a}\succ \overline{b}$ is called
shift-minimal winning. Similarly, a losing coalition type $\overline{b}$ such that $\overline{a}$ is winning
for all $\overline{a}\succ \overline{b}$ are called shift-maximal losing. Each complete simple game can be uniquely
described by either its set of shift-minimal winning coalition types or its set of shift-maximal losing
coalition types.

Based on this insight, Carreras and Freixas (\cite{CaFr96}  pp. 148-150) provided a classification
theorem for complete simple games that allow to enumerate all these
games up to isomorphism by listing the possible values of certain
invariants. Indeed, to each complete simple game $(N,\mathcal{W})$ one can associate the vector $\overline{n} \in \Bbb{N}^t$ as defined above and the list of shift-minimal winning coalitional vectors:
$\overline{m}_p = (m_{p,1},m_{p,2}, \dots, m_{p,t})$ for $1 \leq p \leq r$.

Recall that two simple games $(N,\mathcal{W})$ and $(N^{\prime },\mathcal{W}^{\prime })$
are said to be \emph{isomorphic} if there exists a bijective map $%
f:N\rightarrow N^{\prime }$ such that $S\in \mathcal{W}$ if only if $f(S)\in
\mathcal{W}^{\prime }.$

\begin{theorem} \label{theo:CF96} (Theorem 4.1 in~\cite{CaFr96}) (a) Given a vector $\overline{n} \in \Bbb{N}^t$ and a
matrix $\mathcal{M}$ whose rows $\overline{m}_p = (m_{p,1},m_{p,2},
\dots, m_{p,t})$ for $1 \leq p \leq r$ satisfy the following
properties:

\noindent (i) $0 \leq \overline{m}_p \leq \overline{n}$ for $1 \leq p
\leq r$;
\newline
(ii) $\overline{m}_p$ and $\overline{m}_q$ are not
$\succeq$--comparable if $p \neq q$; i.e., $ \overline{m}_p \bowtie \overline{m}_q$
\newline (iii) if $t=1$, then
$m_{1,1}>0$; if $t>1$, then for every $k<t$ there exists some $p$
such that $$ m_{p,k}>0, \; m_{p,k+1}<n_{k+1};$$ and \newline (iv)
$\mathcal M$ is lexicographically ordered by partial sums, if
$p<q$ either $m_{p,1}>m_{q,1}$ or there exists some $k \geq 1$ such that $m_{p,k} >  m_{q,k}$ and
$ m_{p,i} =  m_{q,i}$ for $h<k$.

Then, there exists a complete simple
game $(N, \mathcal W)$ associated to $(\overline{n},
\mathcal{M})$.

\noindent  (Theorem 4.2 in~\cite{CaFr96}) (b) Two complete games $(N, \mathcal{W})$ and
$(N',\mathcal{W}')$ are isomorphic if {\it and only if} $\overline n
= \overline {n'}$ and $\mathcal M = \mathcal M '$.
\end{theorem}

The pair $(\overline n, \mathcal M)$ is referred as the {\it
characteristic invariants} of game $(N, \mathcal W)$. The authors
prove that these parameters determine the game in the sense that one
is able to define a unique up to isomorphism complete simple game
which possesses these invariants. The characteristic invariants
allow us to count and generate all these games for small values of
$n$. Other applications of the characteristic invariants are to
considerably reduce the calculus of some solutions, as values or
power indices, of the game (see e.g., \cite{FrPu98} for the nucleolus~\cite{Sch69}) or to
study whether a game admits a representation as a weighted
game by studying the consistency of a system of inequalities as we will see below.

If matrix {$\mathcal M$} has only one row, i.e. a unique shift-minimal coalitional vector, then the characteristic invariants reduce to the couple $(\overline n, \overline m)$ with
$$
\begin{array}{ll}
  &1\leq m_{1}\leq n_{1}\\
  &1\leq m_{k}\leq n_{k}-1 \quad \mbox{if} \; 2\leq k\leq t-1, \\
  &0\leq m_{t}\leq \, n_{t}-1,
\end{array}
$$
where the first subindex in matrix $\mathcal M$ is omitted.
It is said, see~\cite{FrPu98}, that $(\overline n, \overline m)$ is a
\emph{complete game with minimum}.

We sketch here how to obtain the characteristic invariants
$(\overline{n}, \mathcal M)$ for the complete game
from winning coalitions and reciprocally.

Given a simple game $(N, \mathcal W)$, for each coalition $S$ we
consider \emph{the vector or coalitional type}
\[
\overline{s}=(\left| S\cap N_{1} \right|
,...,\left| S\cap N_{t}\right| ),
\]
in $\Lambda (\overline{n})$ where $N_{i}$ are the equivalence classes
with $N_{1}>...>N_{t}.$ The vector $\overline n$ is
$(\left|N_{1}\right|  ,...,\left| N_{t}
\right| ).$
The rows of matrix $\mathcal M$ are those $\overline
s$ such that any $S$ is a shift-minimal winning coalition
in the lattice $(\Lambda (\overline{n}),\succeq ).$ Observe that
each vector of indices that $\succeq$--dominates a row of $\mathcal M$
corresponds to winning coalitions.

Conversely, given $(\overline{n},\mathcal M)$ the game
$(N,\mathcal W )$ can be reconstructed, up to isomorphism, as
follows. The cardinality of $N$ is $n=\sum_{i=1}^t n_i$,
the elements of $N$ are denoted by $\{1,2,\dots,n\}.$ The
equivalent classes of $(N, \mathcal W)$ are $N_{1}=\{1,\dots,n_1\},$
$N_2=\{n_1+1,\dots,n_1+n_2\},$ and so on.

Each $S\subseteq N$ with vector $\overline{s}
=(\left| S\cap N_{1}\right|  ,...,\left|
S\cap N_{t}\right| )$ is a winning coalition if
$\overline{s}\, \succeq \,\overline{m}$ for some $\overline{m}$ being a row of $\mathcal M$. Hence, the set of winning coalitions is
\[
\mathcal{W}=\{S\subseteq N \,:  \, \,\overline{s}\,\succeq
\,\overline{m}_p, \; \, \mbox{where} \; \overline{m}_p \; \mbox{is a
row of} \; \mathcal{M}\}.
\]
Notice that a \emph{winning vector} is a vector $\overline{r}$ such that the
coalition representative $R$ is winning. In particular, the
shift-minimal winning coalitions are those with a vector being a row of $\mathcal M$. Precisely,
\[
\mathcal{W}^{s}=\{S\subseteq
N:\overline{s}\,=\,\overline{m}_p \ \text{for some} \ p=1, \dots r\}.
\]
Analogously, one can define the coalitional types of shift--maximal losing coalitions which can be written as rows in a matrix $\mathcal Y$ lexicographically ordered, as requested also for $\mathcal M$, to preserve uniqueness. These coalitional types are the maximal vectors which are not $\succeq$-comparable among them and do not dominate by $\succeq$ any row of $\mathcal M$.

Some particular forms of the pair $(\overline n, \mathcal M)$
reveal the presence of players being either vetoers or nulls. For
instance, if $m_{p,t}=0$ for all $p=1, \dots,r$ the game has $n_{t}$ null players. If
$m_{p,1}=n_1$ for all $p=1, \dots,r$ the game has $n_{1}$ vetoers.

Using the well known fact that any weighted game
admits\textit{\ normalized} representations, where $i \sim j$ if and only
if $w_{i}=w_{j}$, we will consider from now on,
$w=(w_{1},...,w_{t}),$ the vector of weights to be assigned to
the members of each of the $t$ equivalence classes. Using normalized representations a weighted game may be expressed as $[q;w_1(n_1),\dots,w_t(n_t)]$ in
which repetition of weights is indicated within parentheses and $q$ stands for the quota or threshold.
However, these parentheses will be omitted provided that
$\overline{n}=(n_1, \dots,n_t)$ is a known vector. A complete
simple game, $(N,\mathcal{W})$, is weighted \emph{if and only
if} there is a vector
$w=(w_{1},...,w_{t}),$ such that $%
w_{1}>...>w_{t}\geq 0$, which satisfies the system of
inequalities
\[
(\overline{m}_p-\overline{\alpha}_q )\cdot w>0 \ \mbox{ for all
} \ p=1,2,...,r, \ \ q=1, \dots, s
\]
where $r$ is the number of rows of $\mathcal M$, $s$ the number of rows of $\mathcal Y$, and
$\overline{\alpha}_q$ are the rows of $\mathcal Y$.

Only for $n\geq 6$ there are complete simple games which are not weighted.
The following example is the smallest possible illustration
of a complete simple game with minimum, i.e., with one shift-minimal winning vector, that is not a weighted game. It helps us to understand better this kind of
games, which are extensively used in the next section.

\begin{example}
\begin{enumerate}
\item (Example~\ref{EXAMPLE:SC} revisited)  The characteristic invariants for this example are: $\overline n = (5,10)$ and $\mathcal M = ( 5 \ 4)$. Thus,
    $$ \mathcal W = \{ (5,x)  \in \Lambda(5,10)\, : \, x \geq 4 \} $$
    $$ \mathcal W^m = \mathcal W^s = \{ (5,4) \} $$
    Note also that $\mathcal Y =
    \left(
      \begin{array}{cc}
        5 & 3 \\
        4 & 10 \\
      \end{array}
    \right)
    $
    whose rows are the shift-maximal coalitional types. As shown, this game is weighted. In the next section we will show that to prove this it suffices to verify $k$-invariant trade robustness, where $k$ is $2$.
\item (Example~\ref{EXAMPLE:CC} revisited) The characteristic invariants for this example are: $\overline n = (2,8)$ and $\mathcal M = ( 1 \ 6)$. Thus,
    $$ \mathcal W = \{ (x,y) \in \Lambda (2,8) \, : \, x \geq 1 \ \text{and} \ x+y \geq 7 \} $$
    $$ \mathcal W^m = \{ (2,5), \, (1,6) \} $$
    $$ \mathcal W^s = \{ (1,6) \} $$
    Note also that $\mathcal Y =
    \left(
      \begin{array}{cc}
        2 & 4 \\
        0 & 8 \\
      \end{array}
    \right)
    $
    whose rows are the shift-maximal coalitional types.

Note that Example~\ref{EXAMPLE:CC} is not 2-invariant trade robust since the coalitional type trading transform:
\newline
$<(1,6),(1,6) | (2,4),(0,8) >$ is a certificate for it. Hence, the game is not weighted.
\end{enumerate}
\end{example}

\subsection{Two parameters for complete simple games}

Two parameters for a complete simple game are significant for our studies: $r$ the number of rows of $\mathcal M$ or number 
of shift-minimal coalitional vectors and $t$ the number of equivalence classes of players in the game. The conditions that $\mathcal M$ 
must fulfill are described in Theorem~\ref{theo:CF96}. The question we pose here is the following: Are there some values for $r$ and $t$ 
for which $2$-invariant trade robustness is conclusive? The purpose of Section~\ref{sec_main} is to prove that the posed question has an 
affirmative answer for either $r=1$ (no matter the value of $t$) or $t=2$ (no matter the value of $r$), while in Section~\ref{sec_invariant_further} 
we investigate the remaining cases.

Let us remark that the number of complete and weighted games as a function of $|N|$ up to isomorphisms has been determined for these two parameters. 
We use below the notations $cg(n,\star,r)$, $cg(n,t,\star)$,  $wg(n,\star,r)$, and $wg(n,t,\star)$ depending on whether we consider complete or weighted 
games or parameter $r$ or parameter $t$. The first (trivial) exact counting establishes the number of $k$--out--of--$n$ simple games. Each of such games admits 
$[k;\underbrace{1,1,\dots,1}_n]$ as a weighted representation where $k \in \{1, \dots, n \}$. As $t=1$ implies $r=1$ we have
$cg(n,1,\star)=wg(n,1,\star)=n$.

For $r=1$, we have $cg(n,\star,1)=2^n-1$ (see~\cite{FrPu08}) complete simple games with minimum with $n$ players up to isomorphism and the number of 
weighted games with minimum, $wg(n,\star,1)$, is given by
$$  wg(n,\star,1) = \left\{
                        \begin{array}{ll}
                          2^n-1, & \hbox{if} \; n \leq 5 \\
                          \dfrac{n^4-6n^3+23n^2-18n+12}{12}, & \hbox{if} \; n \geq 6 \\
                        \end{array}
                      \right.
$$
cf.~\cite{FrKu14ANOR}.

For $t=2$ we have the nice formula $cg(n,2,\star) = F(n+6)-(n^2 + 4n +8)$ (cf.~\cite{FMR12}) where $F(n)$ are the Fibonacci numbers which constitute a well--known sequence of integer
numbers defined by the following recurrence relation:
$F(0)=0$, $F(1)=1$, and $F(n)=F(n-1)+F(n-2)$ for all $n>1$. Quite curiously the addition of trivial voters, as null voters or vetoers, in complete games with two equivalence classes formed by non-trivial voters give new larger Fibonacci sequences (cf.~\cite{FrKu13EJOR}). Up to now there is not a known formula for
$wg(n,2,\star)$ although it has been proved in~\cite{FrKu14MSS} that $wg(n,2,\star) \leq \frac{n^5}{15} + 4n^4$.

Concerning general enumeration for simple, complete and weighted games it should be said that in the successive works by Muroga et al.~\cite{MTT61, MTK62, MTB70} the number of such games was determined up to eight voters. Only the numbers of complete and weighted games for
$n=9$ voters have been determined since then, cf.~\cite{FrMo10OMS} for the number of complete games for $n=9$ and cf.~\cite{KaKuRiSch15,Ku12} for the number of weighted games for $n=9$.
An asymptotic upper bound for weighted games is given in~\cite{KKZ14} and an asymptotic lower bound for complete games in~\cite{PeSi85}.

\section{Cases for which the test of $\mathbf{2}$-invariant trade robustness is conclusive}
\label{sec_main}

Note first that each simple game with a unique equivalence class of voters, $t=1$, is anonymous
(symmetric), and thus weighted. Non-complete games are not swap robust and therefore they are obviously
not weighted. Hence, we can limit our study to complete simple games.

Prior to study them let us consider the null effect on invariant trade robustness of removing either null or veto players in a given complete simple game. Since adding and removing null players does not change a coalition from winning to losing or the other way round, we can state:

\begin{lemma}
  \label{Lreduction1}
  Let $G$ be a complete simple game and $G'$ be the game arising from $G$ by removing its null players.
  With this we have that $G$ is $m$-invariant trade robust if and only if $G'$ is $m$-invariant trade robust.
\end{lemma}

And a similar result, not as immediate, concerns veto players.

\begin{lemma}
  \label{Lreduction2}
  Let $G$ be a complete simple game and $G'$ be the game arising from $G$ by removing its veto players.
  If $G'$ is a simple game, then $G$ is $m$-invariant trade robust if and only if $G'$ is $m$-invariant trade robust.
\end{lemma}

\begin{proof}
If veto players are present, then each winning coalition of a simple game must contain all veto players. So, in any $m$-trade every
involved losing coalition must also contain all veto players. Let $G=(N,\mathcal{W})$ be a simple game, where
$\emptyset\neq V\subseteq N$ is the set of veto players. If $V=N$ the game $G$ is the unanimity game and therefore weighted. Otherwise we can consider
$G'=(N',\mathcal{W}')$, where $N'=N\backslash V$ and $N'\supseteq S\in \mathcal{W}'$ if and only if $S\cup V\in\mathcal{W}$.
If $\emptyset \in \mathcal{W}'$, then the players in $N\backslash V$ are nulls in $G'$ and the game is indeed weighted. Otherwise
$G'$ is a simple game too. If $G$ is complete, then $G'$ is complete too, see e.g.\ \cite{FrKu13EJOR}. Given an $m$-trade
for $G'$, we can obtain an $m$-trade for $G$ by adding $V$ to all coalitions. For the other direction removing all veto players turns an $m$-trade for $G$ into an $m$-trade for $G'$.
\end{proof}

\subsection{The $\mathbf{2}$-invariant characterization for $\mathbf{r=1}$}

\begin{theorem} \label{t:r=1}
  Each complete simple game $G$ with $r=1$ shift-minimal winning coalition type is either weighted or not $2$-invariant trade
  robust.
\end{theorem}

\begin{proof}
Due to Lemma~\ref{Lreduction1} and Lemma~\ref{Lreduction2} we can assume that $G$ contains neither nulls nor vetoers,
since also the number of shift-minimal winning coalition types is preserved by the transformations used in the respective
proofs.

For $t \geq 3$ types of players let the invariants
of $G$ be given by $\overline{n}=(n_1,\dots,n_t)$ and $\mathcal{M}=\begin{pmatrix}m_1&\dots&m_t\end{pmatrix}$, where
we abbreviate the unique shift-minimal winning coalitional vector by $\overline{m}$. From the
conditions of the general parametrization theorem in \cite{CaFr96} we conclude $1\le m_1\le n_1$, $0\le m_t\le n_t-1$, and
$1\le m_i\le n_i-1$ for all $1<i<t$. If $m_1=n_1$ then $G$ contains veto players and if $m_t=0$ then $G$ contains
null players (cf.~\cite{FrKu13EJOR}). So, we have $1\le m_i\le n_i-1$ for all $1\le i\le t$ in our situation.
We can easily check that $\overline{a}=(m_1-1,m_2+1,m_3+1,m_4,\dots,m_t)$ and $\overline{b}=(m_1+1,m_2-1,m_3-1,m_4,\dots,m_t)$
are losing. Thus, $<\overline{m},\overline{m};\overline{a},\overline{b}>$ is a $2$-trade and $G$ is not $2$-invariant trade
robust.

For $t = 2$ types of players let the invariants
of $G$ be given by $\overline{n}=(n_1,n_2)$ and $\mathcal{M}=\begin{pmatrix}m_1 \ m_2\end{pmatrix}$, where again
we abbreviate the unique shift-minimal winning coalitional vector by $\overline{m}$. From the
conditions of the general parametrization theorem in \cite{CaFr96} we conclude $1\le m_1\le n_1$ and $0\le m_2\le n_2-1$. If $m_1=n_1$ then $G$ contains veto players and if $m_t=0$ then $G$ contains
null players. So, we have $1\le m_i\le n_i-1$ for all $1\le i\le 2$ in our situation.

If $2\le m_2\le n_2-2$, then $\overline{a}=(m_1-1,m_2+2)$ and $\overline{b}=(m_1+1,m_2-2)$
are losing. Thus, $<\overline{m},\overline{m};\overline{a},\overline{b}>$ is a $2$-trade and $G$ is not $2$-invariant trade
robust.

If $m_2=1$ or $m_2=n_2-1$, then both games are weighted.

Indeed, if $m_2=1$, then $\mathcal Y = \left(
                                                          \begin{array}{cc}
                                                            m_1 & 0 \\
                                                            m_1-1 & n_2 \\
                                                          \end{array}
                                                        \right)$,
and the weights $(w_1,w_2)=(n_2,1)$ may be assigned to players in each class respectively, so that a
quota of $m_1 \cdot w_1 + w_2 = m_1\cdot n_2+1$ separates weights of winning and losing coalition types.

If $m_2=n_2-1$, then  $\mathcal Y = \left(
                                                          \begin{array}{cc}
                                                            c_1 & c_2 \\
                                                            m_1-1 & n_2 \\
                                                          \end{array}
                                                        \right)$
where $c_1=\min(n_1,m_1+n_2-2)$ and $c_2=\max(m_1+n_2-2-n_1,0)$.

Now, we have two subcases to consider:

If $c_1=n_1$ a solution is $(w_1,w_2)=(n_1-m_1+2,n_1-m_1+1)$ with quota $q=m_1\cdot w_1+(n_2-1) \cdot w_2 = m_1 \cdot (n_1-m_1+2)+(n_2-1) \cdot (n_1-m_1+1)$.

If $c_1=m_1+n_2-2$ then $c_2=0$ and
a solution is $(w_1,w_2)=(n_2,n_2-1)$ with quota $q=m_1\cdot w_1+(n_2-1) \cdot w_2 = m_1 \cdot n_2+(n_2-1)^2$.
\end{proof}

So, complete simple games with $r=1$ have the property that they are either weighted or not $2$-invariant trade robust. Now
we are going to see that this characterization is also true for $t=2$.

\subsection{The $\mathbf{2}$-invariant characterization for $\mathbf{t=2}$}

Freixas and Molinero~\cite{FrMo09DAM} prove that there is a sequence of
complete simple games $G_m$ with \emph{three} types of equivalent voters, i.e., $t=3$, and \emph{three} types of
shift-minimal winning types, i.e., $r=3$, such that $G_m$ is $m$-invariant trade
but not $(m+1)$-invariant trade robust for each positive integer $m$. Moreover, they state in Conjecture 6.1 of their paper that any complete game with $t=2$ types of equally desirable voters is either weighted or not 2-invariant trade robust. In this subsection we prove this conjecture. Prior to stating the result let us introduce some characterizations for weightedness that will be used in the sequel.  The definition of a weighted game can be rewritten to a quota-free variant:
\begin{lemma}\label{L_weighted_quota_free} Let $G = (N,\mathcal{W})$ be a simple game.
Then,
\newline
$G$ is weighted \; $\Longleftrightarrow$ \;  there are $n$ nonnegative integers $w_1,
\dots, w_n$ such that
\begin{equation} \label{E:InequalityWL}
\sum\limits_{i \in S} w_i > \sum\limits_{i \in T} w_i
\end{equation}
for all $S \in \mathcal W$ and all $T \in \mathcal L$.
\end{lemma}

Moreover, we can use a single weight for equivalent players, i.e., a common
weight $w_i$ for each voter $p \in N_i$ where $N_i$ is an equivalence class of
players according to the desirability relation. If the game is complete we have a total order
among the equivalence classes, $N_1 >\dots>N_t$. Assume from now on $t=2$ so that
$N_1 \neq \emptyset$ and $N_2 \neq \emptyset$ is a partition of $N$. By $\mathcal{W}^v$ we
denote the set of winning coalition types and by $\mathcal{L}^v$ the set of losing coalition types. For instance, $(x,y) \in \mathcal W^v$ means that all coalition $S \subseteq N$ such that $|S\cap N_1|=x$ and $|S\cap N_2|=y$ is winning.
With this,  Lemma~\ref{L_weighted_quota_free} can be rewritten to:
\begin{lemma}
Let $G=(N,\mathcal W)$ be a complete simple game with two types of voters. Then,
\newline
$G$ is weighted \; $\Longleftrightarrow$ \; there are two integers $w_1,w_2\geq 0$ such that
\begin{equation} \label{E:InequalityVectors}
[(x,y) - (x',y')]  \cdot (w_1,w_2)  >  0
\end{equation}
for all   $(x,y) \in \mathcal{W}^{v}$ and all $(x',y') \in \mathcal{L}^{v}$ and ``$\cdot$" stands here for the inner product.
\end{lemma}

For the proof of the theorem for $t=2$ two special parameters of a complete simple game will play a key role so that we give even another reformulation of Lemma~\ref{L_weighted_quota_free}:
\begin{lemma}
\label{L_MP}
Let $G=(N,\mathcal{W})$ be a complete simple game with two types of voters. Then,
\newline
$G$ is weighted \; $\Longleftrightarrow$ \; there are two integers $w_1,w_2\geq 0$ such that
\begin{equation} \label{E:InequalityVectors2}
w_2  >  M w_1\quad\text{and}\quad w_1 >  P w_2,
\end{equation}
where
$$ M = \max_{(x,y) \in \mathcal{W}^v,\,(x',y') \in \mathcal{L}^v\,:\,x' \geq x} \dfrac{x'-x}{y-y'}$$
and
$$ P = \max_{(x,y) \in \mathcal{W}^v,\,(x',y') \in \mathcal{L}^v\,:\,x' < x} \dfrac{y'-y}{x-x'}$$ fulfill
$0\le M<1$ and $P\ge 1$.
\end{lemma}

\begin{proof} Let $(x,y) \in \mathcal{W}^v$ and $(x',y') \in \mathcal{L}^v$. If $x' \geq x$, then
$x+y >x'+y'$, so that $y-y'>x'-x \geq 0$. Thus, $M$ is well defined and we have $0\le M<1$. Also
$P$ is well defined, since we assume $x' < x$ in its definition. For $r=2$ in matrix $\mathcal M$ in Theorem~\ref{theo:CF96} we conclude the existence of a shift-minimal winning type $(a,b)$ with $a>0$ and $b<|N_2|$, i.e.,
$(a-1,b+1)$ is losing. Thus, we have $P\ge \frac{(b+1)-b}{a-(a-1)}=1$.

It remains to remark that all inequalities of the definition of a weighted game are implied by the ones in (\ref{E:InequalityVectors2}).
\end{proof}

\begin{corollary}
\label{C_MP}
Let $G=(N,\mathcal W)$ be a complete simple game with two types of voters. Using the notation from Lemma~\ref{L_MP}, we have
\begin{equation} \label{E:MP}
\text{$G$ is weighted} \ \  \   \Longleftrightarrow \ \ \ M \, P < 1.
\end{equation}
\end{corollary}

We still need an additional technical trivial lemma.
\begin{lemma}
  \label{L_technical_inequality}
  Let $s,u\in\mathbb{R}_{\ge 0}$ and $t,v\in\mathbb{R}_{>0}$. If $t>v$ and $\frac{s}{t}\ge \frac{u}{v}$, then we have
  $
    \frac{s-u}{t-v}\ge \frac{s}{t}
  $.
\end{lemma}


Let us finally prove the result of this subsection, which was previously stated as
Conjecture 6.1 in \cite[page 1507]{FrMo09DAM}.

\begin{theorem} \label{Tmain}
Let $G = (N,\mathcal W)$ be a complete simple game with two types of voters.
Then, $G$ is weighted if and only if $G$ is $2$-invariant trade robust.
\end{theorem}

\begin{proof}
The direct part is immediate since $G$ being weighted implies $G$ satisfies $m$-invariant trade robustness for all $m>1$.
For the other part we start by proving that if $G$ is a complete simple game with $t=2$ types of voters and $G$ is $2$-invariant trade robust, then it is $2$-trade robust.

Let $\left\langle (a_1,b_1),(a_2,b_2);(u_1,v_1),(u_2,v_2)\right\rangle$ be a $2$-trade of $G$
such that $(a_1,b_1)$ and $(a_2,b_2)$ are minimal winning. If both coalition types are shift-minimal, we have finished.
In the remaining cases we construct a $2$-trade with one shift-minimal winning coalition type more than before.
W.l.o.g.\ we assume that $(a_1,b_1)$ is not shift-minimal, so that we consider the shift to $(a_1-1,b_1+1)$.
If $u_1\ge 1$ and $v_1\le n_2-1$ then we can replace $(u_1,v_1)$ by the losing coalitional vector $(u_1-1,v_1+1)$. By symmetry
the same is true for $(u_2,v_2)$. Thus, for the cases, where we can not shift one of the losing vectors, we have
$$
  \left(u_1=0 \,\vee\, v_1=n_2\right) \,\wedge\, \left(u_2=0 \,\vee\, v_2=n_2\right).
$$
\begin{enumerate}
  \item[(1)] $u_1=0$, $u_2=0$:\\
             Since $u_1+u_2=a_1+a_2$ we have $a_1=a_2=0$. Since $(0,b_1)$, $(0,b_2)$ are
             winning and $(0,v_1)$, $(0,v_2)$ are losing, we have $\min(b_1,b_2)>\max(v_1,v_2)$, which
             contradicts $b_1+b_2=v_1+v_2$.
  \item[(2)] $v_1=n_2$, $v_2=n_2$:\\
             Since $b_1+b_2=v_1+v_2$ we have $b_1=b_2=n_2$. Since $(a_1,n_2)$, $(a_2,n_2)$ are winning and
             $(u_1,n_2)$, $(u_2,n_2)$ are losing, we have $\min(a_1,a_2)>\max(u_1,u_2)$, which contradicts
             $a_1+a_2=u_1+u_2$.
  \item[(3)] $u_1=0$, $v_2=n$:\\
             Since $u_1+u_2=a_1+a_2$ we have $a_2\le u_2$. Comparing the winning coalitional vector $(a_2,b_2)$ with the
             losing vector $(u_2,n_2)$, yields $b_2>n_2$, which is not possible.
  \item[(4)] $u_2=0$, $v_1=n$:\\
             Similar to case~(3).
\end{enumerate}
Thus, a shift of one of the losing vectors is always possible, if not both winning vectors are shift-minimal.

According to Theorem~\ref{T:ChowElgotTrade} it remains to prove that for $t=2$ it is not possible for $G$ to be $2$-trade robust but not weighted.

Let $(a,b)$ and $(a',b')$ be two winning vectors, $(c,d)$ and $(c',d')$ be two losing vectors such that
\begin{equation} \label{E:MPextreme}
M = \dfrac{c-a}{b-d} \qquad \text{and} \qquad P = \dfrac{d'-b'}{a'-c'},
\end{equation}
where we assume that the vectors are chosen in such a way that both $c-a$ and $d'-b'$ are minimized.
We remark $a'-c'>0$, $d'-b'>0$, $b-d>0$ and $c-a\ge 0$. The latter inequality can be strengthened to
$c-a>0$, since $c-a$ implies $M=0$ and $MP<1$, which is a contradiction to the non-weightedness of $G$.

Corollary~\ref{C_MP} implies $MP\ge 1$, so that
\begin{equation}
\label{E:MPchanged}
\dfrac{c-a}{b-d} \geq \dfrac{a'-c'}{d'-b'}.
\end{equation}
With this, we have only the following three cases:
\begin{enumerate}
 \item[(a)] $c-a \geq a'-c'$ and $b-d \leq d'-b'$.
 \item[(b)] $c-a > a'-c'$ and $b-d > d'-b'$.
 \item[(c)] $c-a < a'-c'$.
\end{enumerate}
If $c-a = a'-c'$ then we have $b-d \leq d'-b'$ according to Inequality~(\ref{E:MPchanged}), i.e.,
we are in case (a). If $c-a > a'-c'$ then either case~(a) or case~(b) applies. The remaining cases
are summarized in (c).

\begin{enumerate}
 \item[(a)] Since $c+c' \geq a+a'$ and $d+d' \geq b+b'$, we can delete convenient units of some
            coordinates of $(c,d)$ and $(c',d')$ to obtain two well-defined losing vectors
            satisfying $(c'',d'') \leq (c,d)$ and $(c''',d''') \leq (c',d')$ with $c''+c''' = a+a'$
            and $d''+d''' = b+b'$. Thus,
            $
              \langle (a,b), (a',b') ; (c'',d''), (c''',d''')\rangle
            $
            certifies a failure of $2$-trade robustness.
 \item[(b)] Consider $(c'',d'') = (a+a'-c',b+b'-d')$. Since $a'-c'>0$ and $c-a > a'-c'$ we have $a<c''<c$.
            Since $b-d > d'-b'$ and $d'-b'>0$ we have $d<d''<b$. Thus, $(c'',d'')$ is a well-defined
            coalition type. Assuming that $(c'',d'')$ is winning, we obtain
            $$
              \dfrac{c-c''}{d''-d} = \dfrac{\overset{> 0}{\overbrace{c-a}}-(\overset{>0}{\overbrace{a'-c'}})}
              {\underset{>0}{\underbrace{b-d}}-(\underset{>0}{\underbrace{d'-b'}})}\quad
              \overset{\text{Lemma~\ref{L_technical_inequality}}}{\ge} \quad\dfrac{c-a}{b-d} = M,
            $$
            using $b-d > d'-b'$ and Inequality~(\ref{E:MPchanged}). Since $c-c''<c-a$ we have either a
            contradiction to the maximality of $M$ or the minimality of $c-a$. Thus, $(c'',d'')$ has
            to be losing and
            $
              \langle (a,b), (a',b') ; (c',d'), (c'',d'')\rangle
            $
            certifies a failure of $2$-trade robustness.
 \item[(c)] With $c-a < a'-c'$ Inequality~(\ref{E:MPchanged}) implies $d'-b'>b-d$. Consider $(a'',b'')=(c+c'-a,d+d'-b)$.
            Since $c-a> 0$ and $c-a < a'-c'$ we have $c'< a''<a'$. Since $d'-b'>b-d$ and $b-d>0$ we have $b'<b''<d'$. Thus,
            $(a'',b'')$ is a well-defined  coalition type. Assuming that $(a'',b'')$ is losing, we obtain
            $$
              \dfrac{b''-b'}{a'-a''} = \dfrac{\overset{>0}{\overbrace{d'-b'}}-(\overset{>0}{\overbrace{b-d}})}
              {\underset{>0}{\underbrace{a'-c'}}-(\underset{> 0}{\underbrace{c-a}})}
              \quad \overset{\text{Lemma~\ref{L_technical_inequality}}}{\ge} \quad
              \dfrac{d'-b'}{a'-c'} = P
            $$
            using $c-a < a'-c'$ and Inequality~(\ref{E:MPchanged}). Since $b''-b'<d'-b'$ we have either a
            contradiction to the maximality of $P$ or the minimality of $d'-b'$. Thus, $(a'',b'')$ has
            to be winning and
            $
              \langle (a,b), (a'',b'') ; (c,d), (c',d') \rangle
            $
            certifies a failure of $2$-trade robustness.
\end{enumerate}
\vspace*{-5mm}
\end{proof}

Let us have a look at Example~\ref{EXAMPLE:CC} again. We have already observed that this game is not weighted. Nevertheless it can be represented
as the intersection $[7;1,1,1,1,1,1,1,1,1,1]\cap[12;6,6,1,1,1,1,1,1,1,1]$, i.e., there are only two types of provinces -- the large ones, Ontario
and Quebec, and the small ones, see \cite{FMR12}. Indeed the game is complete and the minimal winning vectors are given by $(2,5)$ and
$(1,6)$. The maximal losing vectors are given by $(2,4)$, $(1,5)$, and $(0,8)$, so that we have $M=\frac{1}{2}$ and $P=3$. These values
are uniquely attained by the coalition types $(1,6)$, $(2,5)$ and $(2,4)$, $(0,8)$. Thus we are in case (c) of the proof of Theorem~\ref{Tmain}
and determine the winning coalitional vector $(a'',b'')=(1,6)$. Indeed,
$
  \langle (1,6), (1,6) | (2,4), (0,8)\rangle
$
certifies a failure of $2$-trade robustness. We remark that our previous argument for non-weightedness was exactly of this form and that the coalition 
type $(1,6)$ is shift-minimal. Let us finally conclude this subsection by recalling that Theorem~\ref{Tmain} establishes that a complete game is weighted 
if and only it is $2$-invariant trade robust. Checking that property requires fewer computations than $2$-trade robustness, which was proved 
in~\cite{Her11} to be sufficient for testing weightedness.

\section{Further invariant trade characterizations}
\label{sec_invariant_further}

We have seen in the previous section that complete simple games with $t=2$ or $r=1$ have the property that they are either weighted or not $2$-invariant trade robust.

For other combinations of $r$ and $t$ it is interesting to ascertain which is the maximum integer $m$ such that $m$-invariant trade robustness for the given game with parameters $r$ and $t$ guarantees that it is weighted. Note first that $t=1$ implies $r=1$ so that the pairs $(r,t)=(r,1)$ for $r>1$ are not feasible.
The results in the previous section allow us to conclude that for $(r,t) =(1,t)$ with $t$ arbitrary or for $(r,t) =(r,2)$ with $r$ arbitrary such an $m$ is given by $2$.

The existence of a sequence of complete games being $m$-invariant trade robust but not $(m+1)$-invariant trade robust is proven for $m\geq 4$ by using complete games with parameters $(r,t)=(3,3)$ in~\cite{FrMo09DAM}.  This sequence of games is uniquely characterized by
  $\overline n = (2,m,m-1)$ and
  $
    \mathcal M =
    \left(
      \begin{array}{ccc}
        2 & 0 & 1   \\
        1 & 0 & m-1 \\
        0 & m & m-2
      \end{array}
    \right)
  $. We wonder what is happening for the remaining cases of the parameters $r$ and $t$.

Consider first the smallest case: $(r,t)=(2,3)$.
\begin{lemma}
  \label{lemma_t2r2_itr}
  For $m\ge 3$ the sequence of complete simple games uniquely characterized by
  $\overline n = (2,m,m)$ and
  $
    \mathcal M =
    \left(
      \begin{array}{ccc}
        2 & 0 & 1   \\
        1 & 1 & m-1 \\
      \end{array}
    \right)
  $
  is $(m-1)$-invariant trade robust but not $m$-invariant trade robust.
\end{lemma}

\begin{proof}
For brevity we set $w_1=(2,0,1)$ and $w_2=(1,1,m-1)$. The maximal losing coalition types are given by
$l_1=(2,0,0)$, $l_2=(1,0,m)$, $l_3=(1,1,m-2)$, and $l_4=(0,m,m)$. Since $m\cdot w_2=1\cdot l_1+(m-2)\cdot l_2+1\cdot l_4$,
the game is not $m$-invariant trade robust.

Now assume that there are non-negative integers $a,b,c,d,e,f$ with $a+b = d+e+f>0$ and
$$
  a\cdot w_1+b\cdot w_2\le c\cdot l_1+d\cdot l_2+e\cdot l_3+f\cdot l_4.
$$
We conclude
\begin{eqnarray}
  2a+b &\le& 2c+d+e,\label{eq_1}\\
  b &\le& e+m\cdot f \label{eq_2},\text{ and}\\
  a+(m-1)\cdot b &\le& (m-1)\cdot (d+e+f) \label{eq_3}.
\end{eqnarray}
Assuming $f=0$, we conclude $b\le e$ from Inequality~(\ref{eq_2}), so that we have $a\ge c+d$ due to $a+b = d+e+f$.
Inequality~(\ref{eq_1}) then yields $a=c$, $b=e$, and $d=0$. By inserting this into Inequality~(\ref{eq_3}), we conclude
$c=e=0$, which contradicts $d+e+f>0$. Thus, we have $f\ge 1$.

Inequality~(\ref{eq_1}) yields $c\ge a+f\ge 1$. Assuming $b\le d+e+f$ we conclude $a\ge c$ from $a+b = d+e+f$, which
is a contradiction to $c\ge a+f$ and $f\ge 1$. Thus, we have $b\ge d+e+f+1$.

Inequality~(\ref{eq_3}) yields
\begin{equation}
  \frac{d+f-e}{m-1}-a\ge 1,
\end{equation}
so that $d+f\ge m-1$. Since $c\ge 1$, we have $c+d+e+f\ge m$, i.e., the game is $m-1$-invariant trade robust.
\end{proof}

We remark that the smallest complete simple game with $t=3$, $r=2$ being $3$-invariant trade robust, but not $4$-invariant
trade robust, is given by $\overline n = (2,2,3)$ and
$\mathcal M = \left(
                                              \begin{array}{ccc}
                                                2 & 1 & 0 \\
                                                1 & 0 & 3 \\
                                              \end{array}
                                            \right)
$, as already observed in \cite{FrMo09DAM}. A certificate for a failure of $4$-invariant
trade robustness is given by $\langle w_1,w_1,w_2,w_2;l_1,l_1,l_1,l_2\rangle$, where $w_1=(2,1,0)$, $w_2=(1,0,3)$, $l_1=(2,0,1)$,
and $l_2=(0,2,3)$. The smallest complete simple game with $t=3$, $r=2$ being $4$-invariant trade robust, but not $5$-invariant
trade robust, is attained by Lemma~\ref{lemma_t2r2_itr} for $m=5$.

\begin{lemma}
  \label{lemma_t2r4_itr}
  For $m\ge 3$ the sequence of complete simple games uniquely characterized by
  $\overline n = (2,m,m)$ and
  $
    \mathcal M =
    \left(
      \begin{array}{ccc}
        2 & 1 & 0   \\
        2 & 0 & 2   \\
        1 & 0 & m   \\
        0 & m & m-1 \\
      \end{array}
    \right)
  $
  is $m$-invariant trade robust but not $(m+1)$-invariant trade robust.
\end{lemma}

\begin{proof}
For brevity we set $w_1=(2,1,0)$, $w_2=(2,0,2)$, $w_3=(1,0,m)$, and $w_4=(0,m,m-1)$. The maximal losing coalition types are
given by $l_1=(2,0,1)$, $l_2=(1,0,m-1)$, $l_3=(1,1,m-3)$, $l_4=(0,m,m-2)$, and $l_5=(0,m-1,m)$. Since
$(m-1)\cdot w_1+2\cdot w_3=m\cdot l_1+1\cdot l_5$,
the game is not $(m+1)$-invariant trade robust.

Now assume that there are non-negative integers $a_1,a_2,a_3,a_4,b_1,b_2,b_3,b_4,$ and $b_5$ with $\sum_{i=1}^4 a_i = \sum_{i=1}^5 b_i>0$ and
{
\begin{equation}
  \label{eq_2sum}
  k= \sum_{i=1}^4 a_i\cdot w_i \le \sum_{i=1}^5 b_i \cdot l_i.
\end{equation}
}
It suffices to consider the cases where $k\le m$. We conclude
{ \scriptsize
\begin{eqnarray}
  2a_1+2a_2+a_3&\!\!\!\!\le\!\!\!\!& 2b_1+b_2+b_3,\label{eq_21}\\
  a_1+ma_4&\!\!\!\!\le\!\!\!\!& b_3+m(b_4+b_5)-b_5,\text{ and}\label{eq_22}\\
  2a_2+ma_3+(m-1)a_4 &\!\!\!\!\le\!\!\!\!& b_1+(m-1)(\sum_{i=2}^5 b_i)-2b_3-b_4+b_5 \label{eq_23}.
\end{eqnarray}
}
Let us first assume $a_4=0$. Using Inequality~(\ref{eq_2sum}) and Inequality~(\ref{eq_21}) we obtain
\begin{equation}
  \label{eq_24}
  a_3\ge b_2+b_3+2b_4+2b_5.
\end{equation}
Inserting this into Inequality~(\ref{eq_23}) yields after rearranging
\begin{equation}
  2a_2 + b_2+3b_3+(m+2)b_4+mb_5 \le b_1.
\end{equation}
Since $k\le m$, we have $b_4=b_5=0$. (For $k=m+1$ we have the solution $b_1=m$, $b_2=b_3=b_4=0$, $b_5=1$, $a_1=m-1$, $a_3=2$, and
$a_2=a_4=0$.) With this, Inequality~(\ref{eq_22}) simplifies to $b_3\ge a_1$ and Inequality~(\ref{eq_23}) simplifies to
$b_1+(m-1)b_2+(m-3)b_3\ge 2a_2+ma_3$. Twice the first plus the second inequality gives
\begin{equation}
  b_1+(m-1)b_2+(m-1)b_3 \ge 2a_1+2a_2+ma_3.
\end{equation}
Inserting Inequality~(\ref{eq_2sum}) yields
\begin{equation}
  -b_1+(m-3)(b_2+b_3) \ge (m-3)a_3+a_3.
\end{equation}
Using Inequality~(\ref{eq_24}) we conclude $a_3=b_1=0$. Using Inequality~(\ref{eq_24}) again, we conclude $b_2=b_3=0$, which
is a contradiction to $k=b_1+b_2+b_3+b_4+b_5>0$. Thus, we have $a_4\ge 1$ in all cases.

$2m-2$ times Inequality~(\ref{eq_21}) plus twice Inequality~(\ref{eq_22}) plus Inequality~(\ref{eq_23}) minus $3m-2$ times
Inequality~(\ref{eq_2sum}) yields
\begin{equation*}
  ma_1+ma_2+b_2+b_3+a_4\le (m-1)b_1.
\end{equation*}
Since $a_4\ge 1$ and $b_1\in\mathbb{Z}_{\ge 0}$, we have $b_1\ge 1$.

$3m-4$ times Inequality~(\ref{eq_21}) plus $4$ times Inequality~(\ref{eq_22}) plus twice Inequality~(\ref{eq_23}) minus $6m-4$ times
Inequality~(\ref{eq_2sum}) yields
\begin{equation*}
  ma_3\ge 2a_4+2b_1+(m-2)b_2+(m-2)b_3.
\end{equation*}
Since $a_4,b_1\ge 1$ and $a_3\in\mathbb{Z}_{\ge 0}$, we have $a_3\ge 1$.

$m-\frac{3}{2}$ times Inequality~(\ref{eq_21}) plus Inequality~(\ref{eq_22}) plus Inequality~(\ref{eq_23}) minus $2m-2$ times
Inequality~(\ref{eq_2sum}) yields
\begin{equation}
  \label{eq_b5}
  a_2+\frac{a_3}{2}+a_4\le -\frac{b_2}{2}-\frac{3b_3}{2}+b_5.
\end{equation}
Since $a_3,b_4\ge 1$ and $b_5\in\mathbb{Z}_{\ge 0}$, we have $b_5\ge 2$.

$k=a_1+a_2+a_3+a_4\le m$ plus $m$ times Inequality~(\ref{eq_21})  plus Inequality~(\ref{eq_23}) minus $2m+1$ times
Inequality~(\ref{eq_2sum}) yields
\begin{equation}
  \label{eq_a4}
  2a_2-(m+1)a_4\le -2b_2-4b_3-(m+3)b_4-(m+1)b_5+m.
\end{equation}
Since $b_5\ge 2$ and $a_4\in\mathbb{Z}_{\ge 0}$, we have $a_4\ge 2$.

From Inequality~(\ref{eq_b5}) and $a_2,b_2,b_3\ge 0$ we conclude
$b_5\ge \frac{a_3}{2}+a_4$, so that $b_5\ge a_4+1$ due to $a_3\ge 1$ and $b_5\in\mathbb{Z}_{\ge 0}$.
From Inequality~(\ref{eq_a4}) and $a_2,b_2,b_3,b_4\ge 0$ we conclude
$(m+1)a_4\ge (m+1)b_5-m$. Inserting $b_5\ge a_4+1$ finally yields the contradiction $(m+1)a_4\ge (m+1)a_4+1$.
\end{proof}

We remark that the proof of Lemma~\ref{lemma_t2r4_itr} looks rather technical and complicated at first sight. However,
the underlying idea is very simple. We have to show that the parametric ILP given by inequalities (\ref{eq_2sum})-(\ref{eq_23})
and $9$ non-negative integer variables $a_1,\dots,a_4,b_1,\dots,b_5$ has a minimum value of $m+1$ for the target function
$a_1+a_2+a_3+a_4$. By relaxing the integrality conditions we obtain a corresponding linear program. Minimizing a suitable variable yields a fractional lower bound that can be rounded up. The corresponding dual multipliers are used to conclude
the respective lower bounds directly.

\begin{conjecture}
  \label{conj_itr_t3}
  For each $r\ge 5$, i.e. at least $5$ coalitional types of shift-minimal winning coalitions, there exists a sequence $\left(G_m^r\right)_{m\ge 3}$ of complete simple games, such that  $G_m^r$ is $m$-invariant trade
  robust but not $(m+1)$-invariant trade robust.
\end{conjecture}

\begin{lemma}
  \label{lemma_itr_increase_t}
  Let $G=(\overline{n},\mathcal{M})$ be a complete simple game with $t$ types of voters and $r$ shift-minimal winning coalition types, being
  $m$-invariant trade robust, but not $(m+1)$-invariant trade robust for some $m>1$. Then, there exists a complete simple game $G'$ with $t+1$ types of voters and
  $r$ shift-minimal winning coalition types, which is $m$-invariant trade robust, but not $(m+1)$-invariant trade robust.
\end{lemma}

\begin{proof}
Let $\widetilde{m}^1,\dots,\widetilde{m}^r$ denote the rows of $\mathcal{M}$. If $G$ contains nulls, i.e., if $\widetilde{m}^i_t=0$ for all
$1\le i\le r$, we set $\widehat{m}^i_j=\widetilde{m}^i_j$, $\widehat{m}^i_t=1$, $\widehat{m}^i_{t+1}=0$, $\widehat{n}_j=\overline{n}_j$, $\widehat{n}_t=2$, and
$\widehat{n}_{t+1}=\overline{n}_t$ for all $1\le j\le t-1$, $1\le i\le r$. Otherwise we set $\widehat{m}^i_j=\widetilde{m}^i_j$, $\widehat{m}^i_t=1$,
$\widehat{n}_j=\overline{n}_j$, and $\widehat{n}_{t+1}=2$ for all $1\le j\le t$, $1\le i\le r$.

With this, we choose $G'=(\widehat{n},\mathcal{M}')$,
where $\mathcal{M}'$ is composed of the $r$ rows $\widehat{m}^1,\dots,\widehat{m}^r$. We can easily check that $G'$ is indeed weighted.
Let $l=(l_1,\dots,l_{t+1})$ be a losing coalitional vector in $G'$. If $G$ contains no nulls, then $(l_1,\dots,l_{t})$ is a losing vector
in $G$. Otherwise, $(l_1,\dots,l_{t-1},l_{t+1})$ is a losing coalitional vector in $G$. Thus, a possible certificate for the failure of
$m$-invariant trade robustness for $G'$ could be converted into a certificate for the failure of $m$-invariant trade robustness
for $G$ by deleting the $(t-1)$th or $t$th column of the corresponding vectors -- a contradiction. Similarly, we can
convert a certificate for the failure of $(m+1)$-invariant trade robustness for $G$ into a certificate for the failure of
$(m+1)$-invariant trade robustness for $G'$ by inserting ones into the $(t-1)$th or $t$th column of the corresponding vectors.
\end{proof}

\medskip

The same proof is literally valid in the case of trade robustness:

\begin{lemma}
  \label{lemma_tr_increase_t}
  Let $G=(\overline{n},\mathcal{M})$ be a complete simple game with $t$ types of voters and $r$ shift-minimal winning coalition types, being $m$-trade robust, but not $(m+1)$-trade robust for some $m>1$. Then, there exists a complete simple game $G'$ with $t+1$ types of voters and
  $r$ shift-minimal winning coalition types, which is $m$-trade robust, but not $(m+1)$-trade robust.
\end{lemma}

With these results at hand we may prove that larger classes of games according to parameters $r$ and $t$ never reduce the
largest failures of invariant-trade robustness. Table~\ref{T:Table1} summarizes the invariant trade robust test to be used
for a game to determine whether this is weighted. Looking at this table we conclude that $2$-invariant trade robustness is
conclusive exactly for the cases determined in Section~\ref{sec_main} (conjectured values are printed in bold face), while
for others there is no combination of $r$ and $t$ for which some $m>2$ be enough to ensure that the game is weighted.

\begin{table}[!htp]
  \caption{W: weighted; $-$: not possible; 2-I-T-R: either weighted or not 2-invariant trade robust;
$\infty$-I-T-R: there are games being not $m$-invariant trade robust for all $m$; NW: not a weighted game;
conjectured values in {\bf bold face}.}\label{T:Table1}
 \centering
 \vskip 0.2truecm
  \begin{tabular}{|c|c|c|c|c|c|}
  \hline
  $ r \downarrow \,| \, t \rightarrow$ & 1 & 2 & 3 & 4 & $\dots$ \\ \hline
  1 & W & 2-I-T-R & 2-I-T-R & 2-I-T-R & NW \\
 2 & - & 2-I-T-R & $\infty$-I-T-R & $\infty$-I-T-R & $\infty$-I-T-R \\
  3 & - & 2-I-T-R & $\infty$-I-T-R & $\infty$-I-T-R & $\infty$-I-T-R \\
 4 & - & 2-I-T-R & $\infty$-I-T-R & $\infty$-I-T-R & $\infty$-I-T-R \\
  $\dots$  & - & 2-I-T-R & {\bf$\mathbf{\infty}$-I-T-R} & {\bf$\mathbf{\infty}$-I-T-R} & {\bf$\mathbf{\infty}$-I-T-R} \\
  \hline
\end{tabular}
\end{table}

In words, if one wishes to study the class of complete games with a given pair $(r,t)$ then $2$-invariant trade robustness is a very powerful tool to check weightedness for $t\leq 2$ and $r=1$,
but for the rest of combinations $(r,t)$ we need to look at trade-robustness, which is the purpose of the next section.

\section{Further trade characterizations}
\label{sec_further}

It is well known that all simple games with up to $3$ voters are weighted while there are
non-weighted simple games for $n\ge 4$ voters. Restricting the class of simple games to swap robust simple
games, i.e.\ complete simple games, one can state that up to $5$ voters each such game is weighted
while for $n\ge 6$ voters there are non-weighted complete simple games. Going over to $2$-invariant
trade robustness does not help too much. As shown in \cite{FrMo09DAM}, precisely $3$ of the $60$
non-weighted complete simple games with $n=6$ voters are $2$-invariant trade robust but not $3$-invariant
trade robust. For the classical trade robustness the same authors have shown that all $2$-trade robust
complete simple games with up to seven voters are weighted. By an exhaustive enumeration we have shown
that the same statement is true for $n=8$ voters, i.e., there are exactly $2\,730\,164$ weighted games
and the remaining $13\,445\,024$ complete simple games are not $2$-trade robust. As shown in \cite{GvSl11},
there are complete simple games with $n=9$ voters, which are $3$-trade robust but not $4$-trade robust.
The corresponding example, belonging to a parametric family, consists of nine different types of players, i.e.,
no two players are equivalent.

If the number $t$ of types of players is restricted we can obtain tighter weighted characterizations. For $t=1$ the
games are always weighted and for $t=2$ weightedness is equivalent with $2$-trade robustness (or $2$-invariant trade
robustness for complete simple games). Based on this characterization one can computationally determine the number of complete
simple games with two types of voters which are either weighted, i.e.\ $2$-invariant trade robust, or not weighted, i.e.\ not
$2$-invariant trade robust. In \cite{FMR12} this calculation was executed for $n\le 40$ voters. It turns out that the fraction
of non $2$-invariant trade robust complete simple games quickly tends to $1$. An exact, easy-to-evaluate, and exponentially
growing formula for the number of complete simple games with two types of voters is proven in
\cite{FMR12,KuTa13}. From the upper bound $n^5/15+4n^4$, see \cite{FrKu14MSS}, for the
number of weighted games with two types of voters, we can conclude that this is generally true.
We remark that it is not too hard to compute the number of $2$-invariant trade robust complete simple games with $t=2$
for $n\le 200$, see \cite{FrKu13EJOR}, so that we abstain from giving a larger table.

\begin{table}[!htp]
\label{tab:ThreeCols}
 \caption{Classification of complete simple games with three types of up
          to $15$ voters.
          Parameters:
          size ($n$),
          number of complete simple games (\#\emph{CG}),
          number of weighted simple games (\#\emph{WG}),
          number of non $2$-trade robust complete simple games (\#\emph{N-2T}),
          number of non $3$-trade, but $2$-trade, robust complete simple games (\#\emph{N-3T}).
          }
 \begin{center}
 {\scriptsize{
 \begin{tabular}{rrrrrr}
   \multicolumn{1}{c}{$n$} &
   \multicolumn{1}{c}{\#\emph{CG}} &
   \multicolumn{1}{c}{\#\emph{WG}} &
   \multicolumn{1}{c}{\#\emph{N-2T}} &
   \multicolumn{1}{c}{\#\emph{N-3T}} \\
  \hline\\[-.25cm]
   3 & 0 & 0 & 0 & 0 \\
   4 & 6 & 6 & 0 & 0 \\
   5 & 50 & 50 & 0 & 0 \\
   6 & 262 & 256 & 6 & 0 \\
   7 & 1114 & 976 & 138 & 0 \\
   8 & 4278 & 3112 & 1166 & 0 \\
   9 & 15769 & 8710 & 7059 & 0 \\
   10 & 58147 & 22084 & 36063 & 0 \\
   11 & 221089 & 51665 & 169420 & 4 \\
   12 & 886411 & 113211 & 773186 & 14 \\
   13 & 3806475 & 234649 & 3571788 & 38 \\
   14 & 17681979 & 463872 & 17218019 & 88 \\
   15 & 89337562 & 879989 & 88457385 & 188 \\
   16 & 492188528 & 1610011 & 490578137 & 380 \\
   17 & 2959459154 & 2852050 & 2956606348 & 756 \\
 \end{tabular}
 }}
 \end{center}
 \end{table}

For $t=3$ types of voters we have checked by an exhaustive enumeration
that up to $n=10$ voters each complete simple game is either weighted or not $2$-trade robust. For $n=11$ voters we
have the four examples given by $\overline{n}=(3,3,5)$, $\mathcal{M}_1=\begin{pmatrix}2&2&3\\1&2&5\end{pmatrix}$,
$\mathcal{M}_2=\begin{pmatrix}1&1&2\\0&1&4\end{pmatrix}$, $\mathcal{M}_3=\begin{pmatrix}3&3&0\\3&0&4\\2&3&2\\0&3&5\end{pmatrix}$,
and $\mathcal{M}_4=\begin{pmatrix}3&0&0\\2&0&2\\0&3&1\\0&0&5\end{pmatrix}$, which are all $2$-trade robust but not
$3$-trade robust. This resolves an open problem from~\cite{FrMo09DAM}.
We have computationally checked all complete simple games with three types of voters, i.e.\ $t=3$,  and up to $15$~voters, i.e.\ $|N|\le 15$, see
Table~\ref{tab:ThreeCols}. It seems that the number of games which are $2$-trade robust but not
$3$-trade robust grows rather slowly. Indeed, up to $15$ players every $3$-trade robust such game is weighted.

For $t=4$ types and $n=9$ voters there are several complete simple games which are $2$-trade robust but
not $3$-trade robust, e.g.\ the one given by $\overline{n}=(1,2,3,3)$ and $\mathcal{M}=\begin{pmatrix}1&0&1&0\\
0&2&0&1\\0&1&2&0\\0&1&1&2\\0&0&3&2\end{pmatrix}$. For $t=4$ and $n=10$ there are already $120$ complete simple games
which are $2$-trade robust but not $3$-trade robust.

The next cases to look at, are $t=3$ and $r=2$. For both cases we have already presented examples which are $2$-trade robust but not $3$-trade robust.

In the next section we state a conjecture and ask for several questions related to the problem in relation with the two parameters $r$ and $t$ of a complete game.

\section{Open problems} 

Still we found no example which is $3$-trade robust but not $4$-trade robust.

\begin{question}
  Is every $3$-trade robust complete simple game with $t=3$ types of voters weighted?
\end{question}

\begin{question}
  Is every $3$-trade robust complete simple game with $r=2$ shift-minimal winning coalition types weighted?
\end{question}

As a first step into the direction of these two questions, we have looked at the intersection of both classes, i.e.,
complete simple games with $t=3$ and $r=2$. The game corresponding to the previously presented matrices $\mathcal{M}_1$
and $\mathcal{M}_2$ for $n=11$ voters are of this type and can be generalized:

\begin{lemma}
  \label{Lt3r2}
  For each $k_1,k_2,k_3,l\in\mathbb{N}$ the games uniquely characterized by $\overline{n}_1=(n_1,n_2,n_3)$,
  \newline
  $\mathcal{M}_1=\begin{pmatrix}n_1-(l+1)&n_2-1&n_3-(l+2)\\n_1-2(l+1)&n_2-1&n_3\end{pmatrix}$, where
  $n_1=3+k_1+2l$, $n_2=3+k_2$, $n_3=5+k_3+2l$, and $\overline{n}_2=(n_1,n_2,5+2l)$,
  $\mathcal{M}_2=\begin{pmatrix}l+1&1&l+2\\0&1&2(l+2)\end{pmatrix}$ are $2$-trade
  robust but not $3$-trade robust.
\end{lemma}

We skip the easy but somewhat technical and lengthy proof. Having the nice parametrization at hand, we can easily
state the corresponding generating function, whose coefficients in the resulting power series serve to count the number of those games.
$$
  x^{11}\left(\frac{1}{(1-x)^3(1-x^4)}+\frac{1}{(1-x)^2(1-x^4)}\right)=\frac{x^{11}(x-2)}{(1-x)^3(1-x^4)}
$$
counting the number of such examples, i.e., asymptotically there are $\frac{n^3}{24}+O(n^2)$ such games.
\begin{conjecture}
  \label{Ct3r2}
  All $3$-trade robust complete simple games with $t=3$ and $r=2$ are
  weighted. Additionally, the $2$-trade robust but not $3$-trade robust
  games are exactly those from Lemma~\ref{Lt3r2}.
\end{conjecture}

By an exhaustive enumeration we have checked Conjecture~\ref{Ct3r2} up to $n=20$ voters. From the previous results it is not clear whether a small number of types or shift-minimal winning coalition types
allows to restrict the check of $m$-trade robustness to a finite $m$.

\begin{question}
  For which values of $t$ does a sequence $G_k$ of complete simple games with $t$ types of voters exist such that
  $G_k$ is $k$-trade robust but not $(k+1)$-trade robust for all $k\ge 2$?
\end{question}

\begin{question}
  For which values of $r$ does a sequence $G_k$ of complete simple games with $r$ shift-minimal winning coalition types
  exist such that $G_k$ is $k$-trade robust but not $(k+1)$-trade robust for all $k\ge 2$?
\end{question}

Any progress concerning answers for either the conjecture or the questions posed would be of interest. In Table~\ref{T:Table3} we combine the results from Section~\ref{sec_main} with the questions of this section.
For $r=2$ or $t=3$ we have not found any example being $3$-trade robust but not weighted, but this should be checked formally and become conjectures for future work (in Table~\ref{T:Table3} it appears in black).

\begin{table}[htp!]
\caption{W: weighted; $-$: not possible; 2-I-T-R: either weighted or not 2-invariant trade robust;
{\bf $3$-T-R} for small values of $n$ all games are either weighted or not 3-trade robust -- still a conjecture; NW:
not a weighted game; ?: it is not known if some $m>2$ suffices to assert that $m$-trade robustness implies weighted.}\label{T:Table3}
  \centering
  \vskip 0.2truecm
  \begin{tabular}{|c|c|c|c|c|c|}
  \hline
  $ r \downarrow \,| \, t \rightarrow$ & 1 & 2 & 3 & 4 & $\dots$ \\ \hline
  1 & W & 2-I-T-R & 2-I-T-R & 2-I-T-R & NW \\
 2 & - & 2-I-T-R & {\bf $\mathbf{3}$-T-R} & {\bf $\mathbf{3}$-T-R} &{\bf $\mathbf{3}$-T-R} \\
  3 & - & 2-I-T-R & {\bf $\mathbf{3}$-T-R} & ? & ? \\
 4 & - & 2-I-T-R & {\bf $\mathbf{3}$-T-R} & ? & ?\\
  $\dots$  & - & 2-I-T-R & {\bf $\mathbf{3}$-T-R} & ? & ? \\
  \hline
\end{tabular}
\end{table}

\section{Conclusion}
This paper looks at the characterization of threshold functions within the class of switching functions. We have tried to gather results and efforts that have taken place in different areas of study.
The new results presented in this paper have been exposed in the simple game terminology since some significant advances have been held in this area in the last two decades.
To study the main problem we have restricted ourselves to the class of complete games since non-complete games are not swap-robust and therefore not weighted.

For complete games the test of trade robustness can be computationally relaxed to invariant trade robustness. The strongest condition for invariant trade robustness, $2$-invariant trade robustness, is conclusive for deciding if a given complete game is weighted if the complete game has either a unique coalitional type of shift-minimal winning coalitions or two types of equivalent voters. Larger values for the number of shift-minimal winning coalitions or for the number of equivalence classes show that the tests of trade robustness and invariant trade robustness are complementary. We have found some conspicuous examples of non-weighted games being $k$-trade robust (or $k'$-invariant trade robust for some $k' \geq k$) but not $k+1$-trade robust (or not $k'+1$-invariant trade robust). We have incorporated a number of open questions in hopes of others taking up the challenges that we have left where over.


\section*{Acknowledgements}
This research was partially supported by funds from the Spanish Ministry of Economy and Competitiveness (MINECO) and from the European Union (FEDER funds) under grant MTM2015-66818-P (MINECO/FEDER).

\end{document}